\pdfoutput=1

\documentclass{article}

\usepackage{fancyhdr}
\usepackage{times}
\usepackage[utf8]{inputenc}
\usepackage[american]{babel}
\usepackage{amssymb}
\usepackage{amsmath}
\usepackage{amsthm}
\usepackage{listings}
\usepackage{graphicx}
\usepackage{color}
\usepackage{subfig}
\usepackage{xspace}
\usepackage{enumerate}
\usepackage{xargs}
\usepackage{mathtools} 
\usepackage{nicefrac}
\usepackage{dsfont}
\usepackage{algpseudocode}

\newtheorem{theorem}{Theorem}[section]
\newtheorem{lemma}[theorem]{Lemma}

\newtheorem{definition}[theorem]{Definition}

\newcommandx*{\bigO}[2][1=@pkling_false]{\mathcal{O}\ifthenelse{\equal{#1}{small}}{\bigl(#2\bigr)}{\left(#2\right)}}
\newcommandx*{\LDAUomicron}[2][1=@pkling_false]{\mathrm{o}\ifthenelse{\equal{#1}{small}}{\bigl(#2\bigr)}{\left(#2\right)}}
\newcommandx*{\LDAUOmega}[2][1=@pkling_false]{\Omega\ifthenelse{\equal{#1}{small}}{\bigl(#2\bigr)}{\left(#2\right)}}
\newcommandx*{\LDAUomega}[2][1=@pkling_false]{\omega\ifthenelse{\equal{#1}{small}}{\bigl(#2\bigr)}{\left(#2\right)}}
\newcommandx*{\LDAUTheta}[2][1=@pkling_false]{\Theta\ifthenelse{\equal{#1}{small}}{\bigl(#2\bigr)}{\left(#2\right)}}

\newcommandx*{\set}[2][2=@pkling_false]{\left\{#1\ifthenelse{\equal{#2}{@pkling_false}}{}{\;\middle|\;#2}\right\}}

\newcommand*{\MSOPT}{OPT}
\newcommand*{\MSBLOCKOPT}{OPT_{BL}}

\algblockdefx{ForPar}{EndForPar}[1]%
{\textbf{for }#1 \textbf{do in parallel}}%
{\textbf{end for}}

\makeatletter
\xspaceaddexceptions{\check@icr}
\makeatother

\begin{document}

\title{Non-Preemptive Scheduling on Machines with Setup Times\thanks{This work was partially supported by the German Research Foundation (DFG)
within the Collaborative Research Centre ``On-The-Fly Computing'' (SFB 901)}~~\thanks{A conference version of this paper has been accepted for publication in the proceedings of the 14th Algorithms and Data Structures Symposium (WADS). The final publication is available at www.link.springer.com~\cite{mmms}.}}
\author{Alexander M\"acker \and Manuel Malatyali \and Friedhelm Meyer auf der Heide \and S\"oren Riechers \\ [0.4em]
	Heinz Nixdorf Institute \& 	Computer Science Department\\
	University of Paderborn, Germany\\[0.2em]
	\{amaecker, malatya, fmadh, sriechers\}@hni.upb.de}

\date{}

\maketitle

\begin{abstract}
Consider the problem in which $n$ jobs that are classified into $k$ types are to be scheduled on $m$ identical machines without preemption. A machine requires a proper setup taking $s$ time units before processing jobs of a given type. The objective is to minimize the makespan of the resulting schedule.
We design and analyze an approximation algorithm that runs in time polynomial in $n, m$ and $k$ and computes a solution with an approximation factor that can be made arbitrarily close to $\nicefrac{3}{2}$. 
\end{abstract}

\section{Introduction}
In this paper, we consider a scheduling problem where a set of $n$ jobs, each with an individual processing time, that is partitioned into $k$ disjoint classes has to be scheduled on $m$ identical machines.
Before a machine is ready to process jobs belonging to a certain class, this machine has to be configured properly. That is, whenever a machine switches from processing a job of one class to a job of another class, a setup taking $s$ time units is required. Meanwhile a machine is not available for processing.
The objective is to assign jobs (and the respective setup operations) to machines so as to minimize the makespan of the resulting non-preemptive schedule. 

The considered problem models situations where the preparation of machines for processing jobs requires a non-negligible setup time. These setups depend on the classes of jobs to be processed (i.e.\ they are class-dependent), however, the required setup time is class-independet. Also, jobs might not be preempted, e.g.\ because of additional high preemption costs. Possible examples of problems for which this model is applicable are (1) the processing of jobs on (re-)configurable machines (e.g.\ Field Programmable Gate Arrays) which only provide functionalities required for certain operations (or jobs of a certain class) after a suitable setup or (2) a scenario where large tasks (consisting of smaller jobs) have to be scheduled on remote machines and it takes a certain (setup) time to make task-dependent data available on these distributed machines.

Surprisingly, although a lot of research has been done on scheduling with setup times, we are not aware of results concerning the considered model. 
This is due to the fact that the motivation for considering setup times are often related to preemption of jobs, which is not true for our model.
We discuss some results on these alternative models in the following section on related work. 
Thereafter, we discuss some preliminaries and two simple algorithms in Section~\ref{sec:pre} that include a fully polynomial time approximation scheme (FPTAS) for the considered problem if the number $m$ of machines is constant and a greedy strategy yielding $2$-approximate solutions.
Section~\ref{sec:approxalg} presents the main contribution of this paper which is an algorithm whose approximation factor can be made arbitrarily close to $\nicefrac{3}{2}$ with a runtime that is polynomial in the input quantities $n, k$ and $m$. 
Finally, in Section~\ref{sec:online} we introduce an online version where jobs arrive over time and shortly discuss how to turn, employing a known technique, our offline algorithm into an online strategy with a competitiveness arbitrary close to~$4$.

\subsection{Related Work}
The scheduling problem considered in this paper is a generalization of the classical problem of scheduling jobs on identical machines without preemption and in which setup times are equal to $0$.
This problem has been extensively studied in theoretical research and PTASs with runtimes that are linear in the number $n$ of jobs are known for objective functions such as minimizing (maximizing) the maximum (minimum) completion time or sum of completion times \cite{alon98,hochbaum87}. If the number $m$ of machines is constant, even FPTASs exist \cite{horowitz76}. 

When setup times are larger than $0$, the problem is usually refered to as scheduling with setup times (or setup costs). 
It has also been studied for quite a long time and there is a rich literature analyzing different models and objective functions. 
Usually models are distinguished by whether or not setup times are job-, machine- and/or sequence-dependent.
For an overview on studied problems and results in this context the reader is refered to detailed surveys on scheduling with setup times \cite{allahverdi99,potts00}.
We discuss some closely related problems in the following.
In \cite{monma93}, Monma and Potts consider a model quite similar to ours but they allow preemption of jobs and setup times may be different for each class.
They design two simple algorithms, one with an approximation factor of at most $\max\{3/2-1/(4m-4), 5/3-1/m\}$ if each class is small (i.e.\ setup time plus size of all jobs of a class are not larger than the optimal makespan), and a second one with approximation factor of at most $2-1/(\lfloor m/2 \rfloor+1)$ for the general case.
Later, Schuurman and Woeginger \cite{schuurman99} improve the result for the case that each class consists of only one job that, together with its setup time, is not larger than the optimal makespan. 
The authors design a PTAS for the case where all setup times are identical and a polynomial time algorithm with approximation factor arbitrary close to $4/3$ for non-identical setup times.

A closely related problem was also studied in another context by Shachnai and Tamir \cite{shachnai01}.
They design a dual PTAS for a class-constrained packing problem. 
In contrast to the basic bin packing problem, in this variant each item belongs to a class and each bin has an upper bound on the number of different classes of which items might be placed in one bin. 

The dual problem of our scheduling problem was studied by Xavier and Miyazawa and is known as class-constrained shelf bin packing. For a constant number of classes, an asymptotic PTAS is known for this problem \cite{xavier08} as well as a dual approximation scheme \cite{xavier09}, i.e.\ a PTAS for our problem if $k$ is constant. 

Very recently, Correa et al.\ \cite{corea14} studied the problem of scheduling splittable jobs on unrelated machines. Here, unrelated refers to the fact that each job may have a different processing time on each of the machines. In their model, jobs may be split and each part might be assigned to a different machine but requires a setup before being processed. For this problem and the objective of minimizing the makespan they show their algorithm to have an approximation factor of at most $1+\phi$, where $\phi \approx 1.618$ is the golden ratio. 

In \cite{divakaran11}, an online variant of scheduling with setup times is considered.
The authors propose a $O(1)$-competitive online algorithm for minimizing the maximum flow time if jobs arrive over time at one single machine. 

%%%%%%%%%%%%%%%%%%%%%%%%%%%%%%%%%%%%%%%%%%%%%%%%%%%%%%%%%%%%%%%%%%%%%%%%%%%%%%%%
\section{Model \& Notation}
We consider a model in which there is a set $J = \{1, \ldots, n\}$ of $n$ independent jobs (i.e.\ there are no precedence constraints for jobs) that are to be scheduled on $m$ identical machines $M = \{M_1, \ldots, M_m\}$.
Each job $i$ is available at the beginning and comes with a \emph{processing time} (or \emph{size}) $p_{i} \in \mathbb{N}_{>0}$.
Additionally, the job set is partitioned into $k$ disjoint classes $C = \{C_1, \ldots, C_k\}$, i.e. $J = \bigcup_{i=0}^k C_i$ and $C_i \cap C_j = \emptyset$ for all $i \neq j$.
Before a job $j \in C_i$ can be processed on a machine, this machine has to be configured properly and afterwards jobs of class $C_i$ can be processed without additional setups until the machine is reconfigured for a class $C_{i'} \neq C_i$.
That is, a setup needs to take place before the first job is processed on a machine and whenever the machine switches from processing a job $j \in C_i$ to a job $j' \in C_{i'}$ with $C_i \neq C_{i'}$.
Such a setup takes $s \in \mathbb{N}_{>0}$ time units and while setting up a machine, it is blocked and cannot do any processing.

Given this setting, the objective is to find a feasible schedule that minimizes the makespan, i.e.\ the maximum completion time of a job, and does not preempt any job, i.e.\ once the processing of a job is started at a machine it finishes at this machine without interruption.

In the following we refer to the overall processing time of all jobs of a class $C_i$ as its \emph{workload} and denote it $w(C_i) \coloneqq \sum_{j \in C_i} p_j$ and we assume that for all $1 \leq i \leq n$ it holds that $w(C_i) \leq \gamma \MSOPT$ for some constant $\gamma$ and $\MSOPT$ being the optimal makespan.
By abuse of notation, by $w(C_i)$ we sometimes also represent (an arbitrary sequence of) those jobs belonging to class $C_i$.
To refer to the class $C_i$ of a job $j \in C_i$, we use a mapping $c: J \to C$ with $c(j)=C_i$ and we say a job $j \in C_i$ forms an \emph{individual class} if $c^{-1}(C_i) = \{j\}$. 
The processing time of the largest job in a given instance is denoted by $p_{max} \coloneqq \max_{1 \leq i \leq n}(p_i)$.
We say a machine is an \emph{exclusive machine} (of a class $C_i$) if it only processes jobs of a single class (class $C_i$).
%%%%%%%%%%%%%%%%%%%%%%%%%%%%%%%%%%%%%%%%%%%%%%%%%%%%%%%%%%%%%%%%%%%%%%%%%%%%%%%%

\section{Preliminaries \& Warm-Up}
\label{sec:pre}
As a preliminary for our approximation algorithm presented in Section~\ref{sec:approxalg}, we need to know the optimal makespan before we can actually compute a schedule fulfilling the desired approximation guarantee concerning its makespan. 
However, this assumption is feasible and justified by the applicability of a common notion known as $\alpha$-relaxed decision procedure \cite{hochbaum87}.

\begin{definition}
Given an instance $I$ and a candidate makespan $T$, an \emph{$\alpha$-relaxed decision procedure} either outputs \texttt{no} or provides a schedule with makespan at most $\alpha \cdot T$.
In case it outputs \texttt{no}, there is no schedule with makespan at most $T$. 
\end{definition}
Using such an $\alpha$-relaxed decision procedure (that runs in polynomial time) to guide a binary search on an interval $[l,u]$ with $\MSOPT \in [l,u]$,
we directly obtain a polynomial time approximation algorithm with approximation factor $\alpha$. 
We can find a suitable interval containing the optimal makespan by applying a greedy algorithm that provides an interval of length $\MSOPT$ as follows.

\begin{lemma}
There is a greedy algorithm with runtime $O(n)$ and approximation factor at most $2$.
\end{lemma}

\begin{proof}
First, observe that $T \coloneqq \max\left(s+p_{max},\left\lceil\frac{ks + \sum_{j=1}^n p_j}{m}\right\rceil\right)$ gives a trivial lower bound on \MSOPT.
Now, consider the sequence \[w(C_1), s, w(C_2), s, \ldots, w(C_k).\]
Note that the length of this sequence is exactly $(k-1)s + \sum_{j=1}^n p_j < mT$.
Thus, if we split it at points $lT, l \in \mathbb{N}$ into blocks of length $T$, we obtain at most $m$ blocks.
We now transform each of these blocks in such a way that we obtain a feasible schedule for all jobs on $m$ machines.
To do so, we need to add at most one setup at time $0$ on each machine.
In case a job is split, we also add the remaining processing time of this job to the machine it started on and remove it from the machine where it should have finished.
Hence, we obtain a valid schedule $S$ with makespan $S \leq T+s+p_{max}-1 < 2T$ which yields $T \leq \MSOPT \leq S < 2T$.
\end{proof}

For the sake of simplicity, we assume in the following that by means of this approach we have guessed $\MSOPT$ correctly and show how to obtain an effective approximation algorithm. 
Particularly, using the presented algorithm within the binary search framework as an $\alpha$-relaxed decision procedure, provides the final result. 

\subsection{Constant Number of Machines}
As a first simple result we show that the problem is rather easy to solve if the number $m$ of machines is upper bounded by a constant.
For this case we show how to obtain an FPTAS, i.e.\ an approximation algorithm that, given any $\varepsilon >0$, computes a solution with approximation factor at most $1+\varepsilon$ and runs in time polynomial in $n, k$ and $\frac{1}{\varepsilon}$.
First of all, note that it is simple to enumerate all possible schedules.
To do so, sort the set of jobs according to classes.
Let $\mathcal{S}_i, 0 \leq i\leq n$, be the set of all possible (partial) schedules for the first $i$ jobs.
Let $\mathcal{S}_0 = \emptyset$ and $j_1, \ldots, j_k$, be the indices $i$ at which there is a change from a job of one class to one of another in the ordered sequence and $j_1 \coloneqq  1$.
To compute $\mathcal{S}_i$, if $i \neq j_1, \ldots j_k$, consider each schedule in $\mathcal{S}_{i-1}$ and for each possible assignment of job $i$ to a machine for which a setup took place for $i$'s class $c(i)$ put the corresponding schedule into $\mathcal{S}_i$ (if the makespan is not larger than $T$, others can be directly discarded).
If $i = j_l$ for some $1 \leq l \leq k$, first compute all $2^k-1$ possible extensions of schedules in $\mathcal{S}_{i-1}$ by setups for $i$'s class $c(i)$ and then proceed as in the case before.
Obviously, choosing a schedule $S \in \mathcal{S}_n$ with minimum makespan yields an optimal solution.

In order to obtain an efficient algorithm from this straightforward enumeration of all possible schedules, we first define some dominance relation that helps to remove schedules during the enumeration process for which there are other schedules that will be at least as good for the overall instance.

\begin{definition}
After computing $\mathcal{S}_i$, a schedule $S \in \mathcal{S}_i$ is \emph{dominated} by $S' \in \mathcal{S}_i$ if
\begin{itemize}
  \item $S$ and $S'$ have the same makespan on the first $m-1$ machines and the makespan of $S'$ on the $m$-th machine is at most as large and
  \item in case that $c(i) = c(i+1)$, in $S$ and $S'$ the same machines are set up for $i$'s class $c(i)$.
\end{itemize}
\end{definition}
Note that by removing dominated schedules directly after the computation of $\mathcal{S}_i$ and before the computation of $\mathcal{S}_{i+1}$, we may reduce the size of $\mathcal{S}_i$ without influencing the best obtainable makespan computed at the end in $\mathcal{S}_n$.
However, we cannot ensure that the $\mathcal{S}_i$'s have a small size.
Thus, we consider the following rounding, which is applied before the enumeration: Round up $s$ and the size $p_j$ of each job $j$ to the next integer multiple of $\nicefrac{\varepsilon T}{(n+k)}$, where $\varepsilon >0$ defines the desired precision of the FPTAS.
As to any machine we assign at most $n$ jobs and $k$ setups, the rounding may introduce an additive error of at most $\varepsilon \cdot T \leq \varepsilon \cdot \MSOPT$.
Additionally, the rounding helps to make sure that each $\mathcal{S}_i$ is not too large after removing dominated schedules.
Due to our dominance definition, there are at most $\nicefrac{(n+k)}{\varepsilon}$ different makespans that may occur in schedules in $S_i$.
Hence, there are at most $2^m \cdot (\nicefrac{n+k}{\varepsilon})^m$ many schedules in $\mathcal{S}_i$ that are not dominated, thus proving the following theorem.

\begin{theorem}
If the number $m$ of machines is bounded by a constant, there is an FPTAS with runtime $O(n/\varepsilon +n \log n)$.
\end{theorem}
%%%%%%%%%%%%%%%%%%%%%%%%%%%%%%%%%%%%%%%%%%%%%%%%%%%%%%%%%%%%%%%%%%%%%%%%%%%%%%%%
\section{A $(\nicefrac{3}{2}+\varepsilon)$-Approximation Algorithm}
\label{sec:approxalg}
In this section, we present the main algorithm of the paper. 
The outline of our approach is as follows: \\
(1) We first identify a class of schedules that features a certain structural property and show that if we narrow our search for a solution to schedules belonging to this class, we will still find a good schedule, i.e.\ one whose makespan is not too far away from an optimal one. \\
(2) We then show how to perform a rounding of the involved job sizes and further transformations and thereby significantly decrease the size of the search space. \\
(3) Finally, given such a (transformed) instance, it will be easy to optimize over the restricted class of schedules studied in (1) to obtain an approximate solution to any given instance.

\subsection{Block-Schedules}
We start by discussing the question how to narrow our study to a class of schedules that fulfill a certain property and still, be able to find a provably good approximate solution.  
Particularly, we focus on block-schedules, which are schedules satisfying a structural property, and which we define as follows.
\begin{definition}
Given an instance $I$, we call a schedule for $I$ \emph{block-schedule} if for all $1 \leq i \leq m$ the following holds: In the (partial) schedule for the machines $M_1, \ldots, M_i$, there is at most one class of which some but not all jobs are processed on $M_1, \ldots, M_i$.
\end{definition}
\noindent Intuitively speaking, in a block-schedule all jobs of a class are executed in a block in the sense that they are assigned to consecutive machines and not widely scattered.

In order to prove our main theorem about block-schedules, we first have to take care of jobs having a large processing time in terms of the optimal makespan. 
Let $L_i = \{j \in C_i : \frac{1}{2}OPT -s < p_j< \frac{1}{2}OPT \}$ be the set of \emph{large} jobs of class $C_i$ and $H_i = \{j \in C_i : p_j\geq \frac{1}{2}OPT \}$ be the set of \emph{huge} jobs of class $C_i$.
Based on these definitions we show the following lemma. 

\begin{lemma}\label{lem:hugeLargeJobsIntoNewClasses}
With an additive loss of $s$ in the makespan we may assume that 
\begin{enumerate}
\item Each huge job forms an individual class,  \label{it:hugeClass}
\item There is a schedule with the property that all large jobs of class $C_i$ are processed on exclusive machines, except (possibly) one large job $q_i \in L_i$, for each $C_i$, and %\label{it:exclusiveMachines}
\item $q_i = \text{argmin}_{j \in L_i}\{p_j\}$ is the smallest large job in $C_i$ and the machine it is processed on  has makespan at most $\MSOPT$. \label{it:qMachines}
\end{enumerate}
\end{lemma}
\begin{proof}

We prove the lemma by showing how to establish the three properties by transformations of the given instance $I$ and an optimal schedule $S$ for $I$ with makespan $\MSOPT$. 
To establish the first property, transform $I$ into $I'$ by putting each job $j \in H_i$ into a new individual class, for each class $C_i$.
Because any machine processing such a huge job $j$ cannot process any other huge or large job due to their definitions, the transformation increases the makespan of any machine by at most $s$.

Next, we focus on the second property. 
In $S$ no machine can process two large jobs of different classes.
Hence, we distinguish the following two cases: A machine processes one large job or a machine processes at least two large jobs.
We start with the latter case and consider any machine that processes at least two large jobs of a class $C_i$.
Because these two jobs already require at least $2\left\lceil\nicefrac{(OPT+1)}{2}-s\right\rceil+s\geq OPT-s+1$ time units including the setup time, no job of another class can be processed and thus, this machine already is an exclusive machine.
On the other hand, if a machine $M_p$ only processes one large job $j \in C_i$, we can argue as follows.
The machine $M_p$ works on $j$ for at least $\left\lceil\nicefrac{(OPT+1)}{2}\right\rceil$ time units (including the setup).
Thus, the remaining jobs and setups processed by $M_p$ can have a size of at most $\left\lfloor\nicefrac{(OPT-1)}{2}\right\rfloor$.
If there is still another machine processing a single large job of $C_i$, we can exchange these jobs and setups with this large job and both involved machines have a makespan of at most $OPT+s$.
Also, the machine from which the large job was removed does not contain any huge or large jobs anymore ensuring there is no machine where this process can happen twice.
We can repeat this procedure until all (but possibly one) large jobs are paired so that the second property holds since no machine is considered twice.

Finally, to establish the third property, we can argue as follows: 
If the smallest large job $q_i$ is the only large one on a machine in the schedule $S$, we can do the grouping just described without shifting $q_i$ to another machine satisfying the desired bound on the makespan.
If $q_i$ is already processed on a machine together with another large job, we may pair the remaining jobs but (possibly) one (one that is not processed together with another large job on a machine). In case there is such a remaining unpaired job, we finally exchange $q_i$ with the unpaired job.
The resulting schedule fulfills the desired properties.
\end{proof}

We now put the smallest large job $q_i$ of each class $C_i$ into a new individual class. 
Based on the previous result, there is still a schedule with makespan at most $\MSOPT+s$ for the resulting instance.

In the next lemma, we directly deduce that there is a block-schedule with makespan at most $OPT+s$ if we allow some jobs to be split, i.e.\ some jobs are cut into two parts that are treated as individual jobs and processed on different machines. 
To this end, fix a schedule $S$ for $I$ fulfilling the properties of Lemma~\ref{lem:hugeLargeJobsIntoNewClasses}. 
By $\tilde M$ denote the exclusive machines according to schedule $S$ and by $\tilde C_i$ the class $C_i$ without those jobs processed on machines belonging to $\tilde M$. 

\begin{lemma}\label{lem:almostBlockSchedule}
Given the schedule $S$ fulfilling the properties of Lemma~\ref{lem:hugeLargeJobsIntoNewClasses}, 
there is a schedule $S'$ with makespan at most $OPT+s$ with the following properties:
\begin{enumerate}
  \item A machine is exclusive in $S'$ if and only if it belongs to $\tilde M$ and the partial schedule of these machines is unchanged.
  \item When removing the machines belonging to $\tilde M$ and their jobs from $S$, we can schedule the remaining jobs on the remaining machines such that
  \begin{enumerate}
    \item The block-property holds and
    \item only jobs with size at most $\frac{1}{2}OPT-s$ are split.
  \end{enumerate}
\end{enumerate}
\end{lemma}

\begin{proof}
Remove machines belonging to $\tilde M$ and the jobs scheduled on them from the schedule $S$ obtaining $\tilde S$.
We now show that there is a schedule $S'$ with the desired properties.
Similar to \cite{schuurman99} consider a graph $G=(V,E)$ in which the nodes correspond to the machines in $\tilde S$ and there is an edge between two nodes if and only if in $\tilde S$ the respective machines process jobs of the same class.
We argue for each connected component of $G$. 
Let $m'$ be the number of nodes/machines in this component.
Furthermore, let $C'=\{C'_1, \ldots C'_l\}$ be the set of classes processed on these machines without those formed by single huge or large jobs and $H=\{h_1,\ldots ,h_r\}$ be the set of jobs processed on these machines that are either huge jobs or large jobs forming individual classes. 
Note that $r \leq m'$ since all jobs of $H$ must be processed on different machines in $\tilde S$.
By an averaging argument we know $ \MSOPT + s \geq \frac{1}{m'}\left(\sum_{i=1}^l w (\tilde C'_i) + \sum_{i=1}^r w(h_i)+ (l+r+m'-1)s\right)$
and hence,
\begin{equation}
  \sum_{i=1}^l w(\tilde C'_i) + (l-1)s \leq (m'-r)\MSOPT+\sum_{i=1}^r(\MSOPT- w(h_i)-s).
\end{equation}
Consider the sequence $w(\tilde C'_1), s, w(\tilde C'_2), s,\ldots, s, w(\tilde C'_l)$
of length $\sum_{i=1}^l  w(\tilde C'_i) + (l-1)s$ and split it from the left to the right into blocks of length 
$\MSOPT -w(h_1)-s, \ldots, \MSOPT -w(h_r)-s$,
followed by blocks of length $\MSOPT$. 
Note that each block has non-negative length.
By equation (1) we obtain at most $m'$ blocks and by adding a setup to each block and the jobs $h_i$ plus setup to the first $r$ blocks, we can process each block on one machine.

Consequently, if we apply these arguments to each connected component and add the removed exclusive machines again, we have shown that there is a schedule $S'$ with makespan at most $\MSOPT +s$ satisfying the required properties of the lemma.
\end{proof}

Lemma~\ref{lem:almostBlockSchedule} proves the existence of a schedule that almost fulfills the properties of block-schedules, whose existence is the major concern in this section. 
However, it remains to show how to handle jobs that are split as we do not allow splitting or preemption of jobs and how to place exclusive machines belonging to $\tilde M$, which are not taken care of by the previous lemma, into the obtained schedule in order to yield a block-schedule. 

To simplify the description in the following, when we say we place an exclusive machine $M_i$ before machine $M_j$, we think of a re-indexing of the machines such that the ordering of machines other than $M_i$ and $M_j$ stays untouched but now the new indices of $M_i$ and $M_j$ are consecutive. 
Also, a job $j$ is \emph{started} at the machine that processes (parts of) $j$ and has the smallest index among all those processing $j$. 
A class $C_i$ is processed at the end (beginning) of a machine if there is a job $j \in C_i$ that is processed as the last job (as the first job) on $M_j$.

\begin{lemma}
\label{lem:blockSchedule}
A schedule fulfilling the properties of Lemma~\ref{lem:almostBlockSchedule} can be transformed into a block-schedule with makespan at most $\frac{3}{2}OPT$.
\end{lemma}

\begin{proof}
Consider an arbitrary class $C_i$.
We distinguish three cases depending on where the jobs of $C_i$ are placed in the schedule $S'$ according to the proof of the previous lemma.

\begin{enumerate}[{(1)}]
 \item There is a job in $\tilde C_i$ that is split among two machines $M_j$ and $M_{j+1}$. \label{lem:blockSchedule:proof:cases:1}
 \item There is no job in $\tilde  C_i$ that is split.\label{lem:blockSchedule:proof:cases:2}
 \item $\tilde C_i = \emptyset$.\label{lem:blockSchedule:proof:cases:3}
\end{enumerate}

In case (\ref{lem:blockSchedule:proof:cases:1}) there is a job in $\tilde C_i$ that is split, i.e.\ one part is processed until the completion time of $M_j$ and one from time $s$ on by $M_{j+1}$. 
Hence, we can simply place all machines of $C_i$ between $M_j$ and $M_{j+1}$.
Since jobs that are split have size at most $\frac{1}{2}OPT-s$, we can process any split job completely on the machine on which it was started increasing its makespan to at most $\frac{3}{2}OPT$.
We repeat this process as long as there are jobs with property (\ref{lem:blockSchedule:proof:cases:1}) left.
Note that for each class $C_i$, after having finished case (\ref{lem:blockSchedule:proof:cases:1}), there is no split job left. 

In case (\ref{lem:blockSchedule:proof:cases:2}), we distinguish two cases.
If the jobs in $\tilde  C_i$ have an overall size of at most $\frac{1}{2}OPT$ (including setup), there either is no exclusive machine of $C_i$ and hence no violation of the block-property, or we can process the jobs on an exclusive machine of $C_i$ increasing its makespan to at most $\frac{3}{2}OPT$.
If the jobs have an overall size of more than $\frac{1}{2}OPT$, we distinguish whether $\tilde C_i$ is processed at the end or beginning of a machine $M_j$ or not. 
In the positive case, we can simply place any exclusive machines of $C_i$ behind or before machine $M_j$. 
If $\tilde C_i$ is not processed at the end or beginning of a machine $M_j$, there must be a second class $\tilde C_{i'}$ that is processed at the beginning and a third class $\tilde C_{i''}$ that is processed at the end of machine $M_j$. 
Note that consequently the workload of $\tilde C_{i'}$ processed on $M_j$ cannot be larger than $\frac{1}{2}\MSOPT-s$.
We can perform the following steps on the currently considered machine $M_j$:
\begin{enumerate}[{1.}]
 \item Move all jobs from the class $C_{i'}$ that is processed at the beginning of $M_j$ to machine $M_{j-1}$ if $C_{i'}$ is also processed at the end of $M_{j-1}$, thus only increasing the makespan of $M_{j-1}$ by at most $\frac{1}{2}\MSOPT-s$.
 \item Move all other jobs processed before some workload of $C_i$ to one of their exclusive machines, if they exist.
 \item Shift all the workload $w(\tilde C_i)$ to time $0$ on machine $M_j$ and shift other jobs to a later point in time.
 \item Place all exclusive machines of $C_i$ in front of $M_j$.
\end{enumerate}

In case (\ref{lem:blockSchedule:proof:cases:3}), there are only exclusive machines. 
Such machines can simply be placed behind all other machines. 

These steps establish the block-schedule property and no jobs are split anymore. 
Also note that each machine gets an additional workload of at most $\frac{1}{2}\MSOPT-s$ without requiring additional setups.
Thus, the required bound on the makespan holds, proving the lemma.
\end{proof}

\begin{theorem}\label{thm:blockSchedule}
Given an instance $I$ with optimal makespan $\MSOPT$, there is a transformation to $I'$ and a block-schedule for $I'$ with makespan at most $\MSBLOCKOPT \coloneqq \min\{\MSOPT+p_{max}-1, \frac{3}{2}\MSOPT\}$. It can be turned into a schedule for $I$ with makespan not larger than $\MSBLOCKOPT$.
\end{theorem}
\begin{proof}
The bound $\MSBLOCKOPT \leq \frac{3}{2}\MSOPT$ directly follows from Lemma~\ref{lem:blockSchedule} and the fact that there are only transformations performed on instance $I$ by Lemma~\ref{lem:hugeLargeJobsIntoNewClasses}. 
The second bound (which gives a better result if $p_{max} \leq \frac{1}{2}\MSOPT$) follows by arguments quite similar to those used before: 
If $p_{max} \leq \frac{1}{2}\MSOPT$ holds, we skip the transformation of Lemma~\ref{lem:hugeLargeJobsIntoNewClasses}. Additionally, in the proof of Lemma~\ref{lem:almostBlockSchedule} we do not remove exclusive machines (thus, considering all machines). 
Note that, since we skipped the transformation of Lemma~\ref{lem:hugeLargeJobsIntoNewClasses}, the set $H$ is empty. 
Then, it is straightforward to calculate the second bound of $\MSBLOCKOPT \leq \MSOPT+p_{max}-1$.
\end{proof}
\subsection{Grouping \& Rounding}\label{subsec:groupingAndRounding}
In this section, we show how we can reduce the search space by rounding the involved processing times to integer multiples of some value depending on the desired precision $\varepsilon > 0$ of the approximation. 
We assume that the transformations described in previous sections have already been performed. 
In order to be able to ensure that the rounding of processing times cannot increase the makespan of the resulting schedule too much, we first need to get rid of classes and jobs that have a very small workload in terms of $\MSBLOCKOPT$ and $\varepsilon$.
In the following, we use $\lambda > 0$ to represent the desired precision, i.e.\ $\lambda$ essentially depends on the reciprocal of $\varepsilon$. 
We call every job $j$ with $p_j \leq \MSBLOCKOPT/\lambda$ a \emph{tiny job} and every class $C_i$ with $w(C_i) \leq \MSBLOCKOPT/\lambda$ a \emph{tiny class}. 

\begin{lemma}\label{lem:movingTinyJobsToTinyClasses}
Given a block-schedule for an instance $I$, with an additive loss of at most $4\MSBLOCKOPT/\lambda$ in the makespan we may assume that tiny jobs only occur in tiny classes. 
\end{lemma}

\begin{proof}
We prove the lemma by applying the following transformations to each class $C_i$: 
In a first step, we greedily group tiny jobs of class $C_i$ to new jobs with sizes in the interval $[\MSBLOCKOPT/\lambda, 2 \MSBLOCKOPT/\lambda)$.
In a second step, combine the (possibly) remaining tiny grouped job $j \in C_i$ with a size less than $\MSBLOCKOPT/\lambda$, with an arbitrary other job $j' \in C_i$. 
By this transformation we ensure that tiny jobs only occur in tiny classes and it remains to show the claimed bound on the makespan.

First of all, focus on the first step of the transformation and assume that we do not perform the second step.
Let $S$ be the given block-schedule for instance $I$. 
Lemma~2.3 in the work of Shachnai and Tamir~\cite{shachnai01} proves (speaking in our terms) that for the transformed instance there is a schedule $S'$ with makespan of at most $\MSBLOCKOPT+2\MSBLOCKOPT/\lambda$.
The proof also implies that $S'$ is still a block-schedule: 
For each machine $M_j$ it holds that if $M_j$ is configured for class $C_i$ in the new schedule $S'$, it has also been configured for $C_i$ in the original block-schedule $S$.
Thus, if $S$ is a block schedule, so is $S'$ since we do not have any additional setups in $S'$.

Now assume that also the second step of the transformation is carried out 
and consider the block-schedule $S'$ we just proved to exist.
Distinguish two cases, depending on where the tiny grouped job $j \in C_i$, which was paired in the second step, is processed in schedule $S'$.
If $j$ was paired with a job $j'$ and both $j$ and $j'$ are assigned to the same machine in $S'$, the schedule $S'$ already is feasible for the transformed instance (possibly after shifting $j$ and $j'$ such that they are processed consecutively).  
If the paired jobs $j$ and $j'$ are processed on different machines in schedule $S'$, there is a schedule whose makespan is by an additive of at most $2\frac{\MSBLOCKOPT}{\lambda}$ larger than that of $S'$.
To see this, note that in $S'$ this case can happen at most twice per machine (for the classes processed at the beginning and end of the machine).
Hence, we can place any paired jobs $j$ and $j'$ on the same machine yielding a schedule for the transformed instance with the claimed bound on the makespan.
Finally, note that we can easily turn a schedule fulfilling the claimed bound on the makespan into a schedule for the original instance $I$ satisfying the same bound on the makespan. 
\end{proof}

Next, we take care of tiny classes that still might occur in a given instance. 
Again, without losing too much with respect to the optimal makespan we may assume a simplifying property as shown in the next lemma.

\begin{lemma}\label{lem:removingTinyClasses}
With an additive loss of at most $4\MSBLOCKOPT/\lambda$ in the makespan we may assume the following properties:
\begin{enumerate}
  \item Each tiny class consists of a single job. 
  \item In case that $\MSBLOCKOPT/\lambda > s$, it has size $\MSBLOCKOPT/\lambda-s$.
\end{enumerate}
\end{lemma}

\begin{proof}
At first note that with an additive loss of at most $2\MSBLOCKOPT/\lambda$ in the makespan, we may assume that a tiny class is completely scheduled on one machine in a block-schedule. 
This is true because of reasons similarly used in the proof of the previous lemma: 
For each machine it holds that there are at most two different tiny classes of which some but not all jobs are processed on this machine.
Hence, we may shift all jobs of such classes to one machine and thereby increase the makespan by at most $2\MSBLOCKOPT/\lambda$.

Now distinguish two cases depending on whether $\MSBLOCKOPT/\lambda > s$ or not.
If this is the case, determine the length $L$ of the sequence of all tiny classes (including setup times), round up $L$ to an integer multiple of $\MSBLOCKOPT/\lambda$, remove all tiny classes from the instance and instead, introduce $\lambda L/\MSBLOCKOPT$ new classes each comprised of a single job with workload $\MSBLOCKOPT/\lambda-s$.
Observe that, given a block-schedule in which each tiny class is completely scheduled on one machine, we can simply replace tiny classes by these new classes, increasing the makespan by an additive of at most $\MSBLOCKOPT/\lambda$.
Also, this schedule implies a schedule for the instance in which tiny classes have not been grouped and its makespan is by an additive of at most $\MSBLOCKOPT/\lambda$ larger. 
This schedule is simply obtained by again replacing grouped tiny classes by its respective original classes.

In case that $\MSBLOCKOPT/\lambda \leq s$, we simply group all jobs of a tiny class $C_i$ to a new job $j$ of the same size $p_j = w(C_i)$. Due to the fact that we might assume that a tiny class is completely scheduled on one machine, this proves the lemma.
\end{proof}

From now on, we assume that we have already conducted the grouping from the two previous lemmas and we describe how to round job sizes in order to reduce the search space for later optimization. 
The rounding approach is quite common for makespan scheduling.

Given an instance $I$, we compute its rounded version $I'$ by rounding up the size of each job to the next integer multiple of $\nicefrac{\MSBLOCKOPT}{\lambda^2}$.
We know that there is a block-schedule with makespan at most $\MSBLOCKOPT + 8\frac{\MSBLOCKOPT}{\lambda}$ and we also assume that the properties from Lemma \ref{lem:removingTinyClasses} hold. 

In case that $\MSBLOCKOPT/\lambda > s$ each job has either a processing time of at least $\nicefrac{\MSBLOCKOPT}{\lambda}$ or forms a tiny class with workload at least $\nicefrac{\MSBLOCKOPT}{\lambda}-s$.
On the other hand, in case that $\MSBLOCKOPT/\lambda \leq s$ and there are tiny classes consisting of a single job, to execute such a job, we need perform a setup first which yields a processing time of at least $\nicefrac{\MSBLOCKOPT}{\lambda}$ as well.
Hence, we can have at most $\lambda + 8$ jobs on one machine in the considered block-schedule, leading to an additive rounding error of at most $(\lambda+8)\cdot\nicefrac{\MSBLOCKOPT}{\lambda^2}$ in the makespan.
Therefore, by choosing $\lambda$ appropriately, there is a solution to the rounded instance that approximates $\MSBLOCKOPT$ up to any desired precision $\varepsilon >0$.

\subsection{Optimization over Block-Schedules}

We are ready to show how to compute a block-schedule for the rounded instance $I'$ with makespan at most $(1+\varepsilon)\MSBLOCKOPT$ for any $\varepsilon >0$. The obtained schedule directly implies a schedule for the original instance $I$ with the same bound on the makespan. 

\begin{lemma}\label{lem:numberOfClassConfigurations}
 If all job sizes are a multiple of $\nicefrac{\MSBLOCKOPT}{\lambda^2}$ and $\lambda > 0$ is a constant, there is only a constant number $c_{cl}$ of different class-types.
 \end{lemma}
\begin{proof}
 We can represent any class $C_i$ by a tuple of length $\lambda^2$ describing how many jobs of each size $l \cdot \nicefrac{\MSBLOCKOPT}{\lambda^2}$, $1 \leq l \leq \lambda^2$, occur in class $C_i$.
 As each class has a size of at most $\gamma\cdot\MSOPT$, each entry of the tuple is limited by $\gamma\lambda^2$ and there is at most a constant number $c_{cl}\coloneqq(\gamma\lambda^2)^{\lambda^2}$ of different tuples describing the classes of $I'$.
 In the following we say that all classes represented by the same such tuple are of the same \emph{class-type}, proving the lemma. 
\end{proof}

We can represent the classes that have to be scheduled as a tuple of size $c_{cl}$ where each entry contains the number of times classes of the respective class-type occur.
Given a block-schedule $S$, we consider machine configurations that describe which classes are finished on the first $i$ machines.
We denote the sub-schedule induced by these first $i$ machines by $S_i$.

\begin{lemma}
 If all job sizes are a multiple of $\nicefrac{\MSBLOCKOPT}{\lambda^2}$ and $\lambda > 0$ is a constant, the number of machine configurations representing $S_i$ for some block-schedule $S$ and some $i > 0$ is bounded by a value $c_{conf}$ that is polynomial in $m$. 
\end{lemma}
\begin{proof}
 First, note that in a block-schedule $S$, for every $S_i$, there is at most one class that is split due to the block-schedule property.
 Now, to uniquely define a candidate configuration, we need to store information about the classes that are finished, and in case a class has been split, the type of this class and which jobs of this class are finished.
 We reserve $c_{cl}$ entries for the finished classes, where each entry corresponds to the number of classes of the certain type that has been fully finished.
 Each entry is at most $m\cdot(\lambda+8)$ with similar arguments as in the proof of Lemma~\ref{lem:numberOfClassConfigurations} and the reasoning concerning the maximum rounding error.
 For the class that has been split, we store the type of that class in an extra entry, which gives $c_{cl}$ possible values.
 If there is no class that has been split, we leave this entry empty adding another possible value to the entry.
 Finally, we store the number of jobs from the split class that have been finished for each job size as $\lambda^2$ additional entries, where each entry does not exceed $c_{cl} \cdot\lambda$ similar to the structure in Lemma~\ref{lem:numberOfClassConfigurations}.
 Overall, we write a configuation as a tuple $\left(n_1,\ldots,n_{c_{cl}},j,u_{1},\ldots,u_{\lambda^2}\right)$ and thus there are at most $c_{conf}\coloneqq(m(\lambda+8))^{c_{cl}}\cdot (c_{cl} + 1) \cdot (c\lambda)^{\lambda^2}$ possible configurations, which proves the lemma. 
\end{proof}

We now build a graph where we add a node for each machine configuration.
We draw a directed edge from node $u$ to $v$ if and only if the machine configuration corresponding to $v$ can be reached from the configuration $u$ by using at most one additional machine with makespan not larger than  $(1+\varepsilon)\MSBLOCKOPT$.
That is, assuming $u$ is a possible sub-schedule induced by the first $i$ machines, we verify whether $v$ is a possible sub-schedule induced by the first $i+1$ machines.
We can do so as we assume that we have guessed $\MSOPT$ correctly and we can hence determine $\left(1 + \varepsilon\right)\MSBLOCKOPT$ which is the amount of workload we will fit on one machine.
In order to determine the edges of the graph that describes our search space, we prove the following lemma, where we denote $\mathds{1}_B$ as the indicator variable which is $1$ in case the boolean condition $B$ is satisfied and $0$ elsewise.
Also, we define $m_{pk}$ to be the number of jobs of type $k$ in class-type $p$, where $k\in\{1,\ldots,\lambda^2\}$ and $p\in\{1,\ldots,c_{cl}\}$.

\begin{lemma}
 If each configuration $\left(\vec{n},j,\vec{u}\right) = \left(n_1,\ldots,n_{c_{cl}},j,u_1,\ldots,u_{\lambda^2}\right)$ is represented by a node, there is a directed edge from node $V = \left(\vec{n},j,\vec{u}\right)$ to $\tilde{V} = \left(\vec{\tilde{n}},\tilde{j},\vec{\tilde{u}}\right)$ if and only if
 \begin{align}
  &\mathds{1}_{j\neq \tilde{j} \vee u\neq \tilde{u}} s + \sum_{k=1}^{\lambda^2} \left(\left(\tilde{u}_k - u_k\right)\cdot k \cdot \frac{\MSBLOCKOPT}{\lambda^2}\right) \nonumber\\
  &+ \sum_{p=1}^{c_{cl}} \left(\left(\tilde{n}_p - n_p\right)\left(s +  \sum_{k=1}^{\lambda^2} m_{pk} \cdot k \cdot \frac{\MSBLOCKOPT}{\lambda^2}\right)\right) \nonumber \\
  &\leq \left(1+\varepsilon\right)\cdot \min\left\{\MSOPT+p_{max}-1, \frac{3}{2}\MSOPT\right\}. \label{ineq:edgesInGraph}
 \end{align}
\end{lemma}
\begin{proof}
We prove the statement for the following cases:
 \begin{enumerate}
  \item\label{enum:prf:edges:unequalClasses} $j\neq \tilde{j}$:
  
  First, note that the number of classes of type $p\in\{1,\ldots,c_{cl}\}$ that have been completed between node $V$ and node $\tilde{V}$, i.\,e.\ on the additional machine, is expressed in the value $\left(\tilde{n}_p - n_p\right)$.
  Now, in order to finish all jobs from a class of type $p\in\{1,\ldots,c_{cl}\}$, we need to configure the machine for this class and afterward, the workload of all jobs contained in that class type needs to be finished.
  This leads to an overall processing time of $s + \sum_{k=1}^{\lambda^2} m_{pk} \cdot k \cdot \nicefrac{\MSBLOCKOPT}{\lambda^2}$ for all jobs of the specific class type.
  In case the class being finished is $j$, there is still the same setup time, but there is less workload to be completed.
  This can be described by subtracting the amount of work already finished in node $V$, which is $\sum_{k=1}^{\lambda^2} u_k\cdot k \cdot \nicefrac{\MSBLOCKOPT}{\lambda^2}$.
  Additionally, to reach the state represented by node $\tilde{V}$, class $\tilde{j}$ needs to be set up and the workload depicted by $\tilde{u}$ has to be completed yielding an additional processing time of $s + \sum_{k=1}^{\lambda^2} \tilde{u}_k \cdot k \cdot \nicefrac{\MSBLOCKOPT}{\lambda^2}$.
  Summing all these times up, we get exactly the value on the left-hand side of inequality (\ref{ineq:edgesInGraph}).
  \item $j = \tilde{j} \wedge \exists i, u_i > \tilde{u}_i$:
  
  In this case, we indeed have $j = \tilde{j}$, but as we have $u_i > \tilde{u}_i$ for some $i$, there are more jobs of type $i$ finished in $V$ than in $\tilde{V}$.
  Thus, the class that had been partly executed at the end of $V$ needs to be completed and the proof of case \ref{enum:prf:edges:unequalClasses} similarly applies.
  \item $j = \tilde{j} \wedge \forall i, u_i \leq \tilde{u}_i \wedge u \neq \tilde{u}$:
  
  Here, the scheduler does not necessarily need to finish the class that had been partly executed at the end of $V$.
  However, the overall necessary workload is the same whether the work on the current class is only continued and not finished (cost $s$ for setting up the machine for the class) and a new class is fully executed (cost $s$) or whether it is finished (cost $s$) and a new class of the same type is initialized (cost $s$) and not finished.
  Thus, the proof of case \ref{enum:prf:edges:unequalClasses} still applies.
  Note that this also holds if the number of classes of type $j$ that have been fully finished is the same in $V$ and $\tilde{V}$.
  
  \item $j = \tilde{j} \wedge u = \tilde{u}$:
  
  In this case, we save an overall workload of $s$ in comparison to the other cases.
  This is due to the fact that we do not need to perform a setup for class $j$ as we can restrict ourselves to executing entire classes.
 \end{enumerate}
 Combining these cases completes the proof. 
\end{proof}

It is time to show that a schedule using only $m$ machines and finishing all jobs exists.

\begin{lemma}
We can construct a graph $G$ such that there is a path from the node representing no job at all (source) to the node representing the entire instance $I'$ (target) that has a length of at most $m$.
\end{lemma}
\begin{proof}
Using Theorem \ref{thm:blockSchedule}, there is a block-schedule with makespan at most $\MSBLOCKOPT$.
 Due to Lemma \ref{lem:movingTinyJobsToTinyClasses} and Lemma \ref{lem:removingTinyClasses} together with the additive rounding error and a suitable value for $\lambda$ depending on $\varepsilon$, there exists a solution to the rounded instance $I'$ with makespan at most 
 \begin{align*}
    (1+\varepsilon) \MSBLOCKOPT = (1+ \varepsilon) \min\{\MSOPT + p_{max}-1, \nicefrac{3}{2} \MSOPT\}.
 \end{align*}
 By construction, the considered graph must contain a path describing this schedule, proving the lemma.
 Note that this naturally gives an approximation with factor at most $\left(1+\varepsilon\right)\left(\MSOPT+p_{max}-1\right)$ which is better in the case of $p_{max}\leq\nicefrac{1}{2}\MSOPT$ and which gives a PTAS for unit processing times.
\end{proof}

\begin{theorem}
 By using breadth-first search on $G$, we can determine a schedule for the original instance $I$ with makespan at most 
 \[(1+\varepsilon)\min\left\{\frac{3}{2}\MSOPT,\MSOPT+p_{max}-1\right\}.\]
 It implies an algorithm with exactly this approximation guarantee and runtime polynomial in $n,k$ and $m$. 
\end{theorem}
\begin{proof}
 Obviously, if we use breadth-first search on the graph, where the source vertex corresponds to the state where no job has been finished and the target vertex corresponds to the state where all jobs have been finished, this gives a path $p = (v_0, v_1, \ldots, v_l)$ of length at most $m$.
 By following this path and considering the difference between two consecutive nodes $p_{i-1}$ and $p_{i}$, we can efficiently determine the jobs from instance $I'$ to be scheduled on machine $M_i$.
 The resulting schedule can be efficiently transformed back into the final schedule for instance $I$ as already discussed during the description of the transformation we apply to $I$. Also, since the number of nodes is essentially the number of configurations, which in turn is polynomial in $m$, the search can be carried out efficiently. 
\end{proof}

%%%%%%%%%%%%%%%%%%%%%%%%%%%%%%%%%%%%%%%%%%%%%%%%%%%%%%%%%%%%%%%%%%%%%%%%%%%%%%%%

\section{An Online Variant}
\label{sec:online}
While in our original model discussed before we have assumed that all jobs are available at time $0$, also online variants are of fundamental interest. 
Consider a model in which a release time $r_j$ is associated with each job $j$ and a job is not known to the scheduler before $r_j$, i.e.\ jobs arrive in an online fashion.
The objective remains the minimization of the makespan and we assess the quality of an online algorithm using standard competitive analyses:
An online algorithm is $c$-competitive if, for any instance, the makespan of the schedule computed by the online algorithm is by a factor of at most $c$ larger than that of an optimal (offline) solution.

A very simple lower bound on the competitiveness of any online algorithm can be obtained by exploiting the fact that any online algorithm cannot know the class of a job arriving later on in advance and hence, cannot perform a suitable setup operation beforehand. The following lemma shows that this fact results in a lower bound that can be arbitrary close to $2$.

\begin{lemma}
No online algorithm can be $c$-competitive for $c \leq 2 - \varepsilon$ and any $\varepsilon > 0$.
\end{lemma}

\begin{proof}
Consider an instance with (without loss of generality) $m=2$ machines and the following adversary: At time $0$ the adversary releases the first job of some class $C_1$ with processing time $p_1 = 1$. 
Then, at time $s$ a second job with processing time $p_2 =1$ is released belonging to a class for which the online algorithm has not performed a setup yet.
Trivially, the optimal algorithm obtains a schedule with makespan $s+1$ by performing at time $0$ a setup for the first job on one machine and one for the second job on the second machine, and then processes the two jobs until time $s+1$.
Any online algorithm cannot do better than performing a setup for the second job at time $s$ and then processing this job. This directly implies a makespan of at least $2s+1$. Hence, the competitiveness is at least $\frac{2s+1}{s+1}$, which can be arbitrary close to $2$ for large setup times $s$. 
\end{proof}

In \cite{onoff91}, Shmoys et al.\ present a quite general technique to turn an offline algorithm for a scheduling problem without release dates and an approximation factor of $\alpha$ into a $2\alpha$-competitive online algorithm for the respective problems with release dates. 
Although this factor of $2$ does not directly carry over to our scheduling problem since we also have to take into account setup processes, a slight modification yields the following result.

\begin{theorem}
If each job is associated with a release time and jobs are revealed to the scheduler over time at these release times, our algorithm implies a polynomial time $c$-competitive online algorithm and $c$ can be made arbitrarily close to $4$.
\end{theorem}

\begin{proof}
Although the proof is pretty much the same as that given in \cite{onoff91}, for the sake of completeness we state it again. 
Let $0$ be the point in time where the first jobs arrive and call this set of jobs $S_0$. 
We apply our approximation algorithm and obtain a schedule for the jobs in $S_0$ and let $F_0$ be its makespan.
Next we consider those jobs arriving between time $0$ and $F_0$, call the set of them $S_1$ and compute a schedule for $S_1$ that begins at time $F_0$ and ends at time $F_1$.
Generally, we call the set of jobs released during the interval $(F_{i-1}, F_i]$ the set $S_{i+1}$ where $F_i$ is the point in time where the schedule for $S_i$ finishes. Then we schedule $S_{i+1}$ using our approximation algorithm. 

Let $F_l$ be the makespan of the entire schedule. 
We can determine an upper bound on $F_l$ as follows:
First, observe that $F_{l-1} \leq F_{l-2} + (1+\varepsilon)(\MSOPT + p_{max}+s)$ since the approximation quality of our algorithm makes shure that we need at most $(1+\varepsilon)(\MSOPT + p_{max} +s)$ time to process the jobs in $S_{l-1}$. 
Note that we may need the additional setup time $s$ because the optimal schedule might have already performed necessary setups earlier.
Second, consider the instance $I'$ obtained from $I$ by releasing the jobs of $S_l$ at time $F_{l-2}$. 
We observe that $F_l-F_{l-1} \leq (1+\varepsilon)(\MSOPT + p_{max}+s)-F_{l-2}$ by the approximation quality of our algorithm and the fact that also the optimal solution cannot schedule jobs of $S_l$ before $F_{l-2}$.
Putting both inequalities together we obtain $F_l \leq 2(1+\varepsilon)(\MSOPT + p_{max} +s) \leq 4(1+\varepsilon) \MSOPT$, proving the theorem. 
\end{proof}
It remains an interesting question for future work, whether the gap between the lower and the upper bound can be narrowed by more clever lower bound constructions and/or a strategy specifically tailored to the online scenario. 

\end{document}